\documentclass[12pt]{article}
\RequirePackage{amsmath,amssymb}

\usepackage{amsthm}
\usepackage[toc,page]{appendix}
\usepackage{bm}
\usepackage{color}
\usepackage{graphicx}
\usepackage[letterpaper, portrait, margin=1in]{geometry}
\usepackage{dsfont}
\usepackage{natbib} 
\usepackage{thmtools}
\usepackage{xr}
\usepackage{xcite}

\newcommand{\iv}{\mathds{1}}

\newcommand{\Z}{\mathcal{Z}}
\newcommand{\Zset}{\Z}

\newcommand{\bZ}{Z}
\newcommand{\bz}{z}

\newcommand{\bYone}{Y(1)}
\newcommand{\bYzero}{Y(0)}
\newcommand{\etaref}{\eta_0}

\newcommand{\Cz}[1]{C(\hat{\tau}(Z); #1)}
\newcommand{\Czln}[1]{C(\hat{\tau}(z); #1)}

\newcommand{\Lz}[1]{L_\alpha(#1)}
\newcommand{\Uz}[1]{U_\alpha(#1)}

\newtheorem{definition}{Definition}
\newtheorem{proposition}{Proposition}
\newtheorem{theorem}{Theorem}
\newtheorem{corollary}{Corollary}
\newtheorem{example}{Example}

\theoremstyle{remark}

\newcounter{ex1}
\newcounter{ex2}
\newcounter{ex3}
\newcounter{ex4}
\newcounter{ex5}
\newcounter{exnew}

\begin{document}

\renewcommand\thmcontinues[1]{continued}

\title{Conditional As-If Analyses in Randomized Experiments\thanks{\noindent{\textit{Email: } \texttt{nicole.pashley@rutgers.edu}. The authors would like to thank Zach Branson, Tirthankar Dasgupta, Kosuke Imai, and participants of the 2019 Atlantic Causal Inference Conference in Montreal, Quebec as well as the Miratrix C.A.R.E.S. Lab for useful feedback.
Nicole Pashley was supported by the National Science Foundation Graduate Research Fellowship under Grant No. DGE1745303, while working on this paper.
The research reported here was also partially supported by the Institute of Education Sciences, U.S. Department of Education, through Grant R305D150040.
The opinions expressed are those of the authors and do not necessarily represent views of the Institute, the U.S. Department of Education, or the National Science Foundation.
}}}

    \author{Nicole E. Pashley\\Department of Statistics, Rutgers University   \and Guillaume W. Basse\\  Department of Statistics, Stanford University \and Luke W. Miratrix\\ Harvard Graduate School of Education}
    
\maketitle

\begin{abstract}
The injunction to `analyze the way you randomize' is well-known to statisticians since Fisher advocated for randomization as the basis of inference.
  Yet even those convinced by the merits of randomization-based inference seldom follow this injunction to the letter.
  Bernoulli randomized experiments are often analyzed as completely randomized experiments, and completely randomized experiments are analyzed as if they had been stratified; more generally, it is not uncommon to
  analyze an experiment as if it had been randomized differently.
  This paper examines the theoretical foundation behind this practice within
  a randomization-based framework. Specifically, we ask when is it legitimate
 to analyze an experiment randomized according to one design as if
  it had been randomized according to some other design. We show that
  a sufficient condition for this type of analysis to be valid is that the
  design used for analysis be derived from the original design by an appropriate form of conditioning.
  We use our theory to justify certain existing methods, question others, and
  finally suggest new methodological insights such as conditioning on
  approximate covariate balance.
\end{abstract}

\noindent%
{\it Keywords:} Ancillary statistics; Causal inference; Conditional inference; Randomization inference; Relevance.
\vfill





\section{Introduction}\label{sec:intro}

It is a long-standing idea in statistics that the design of an experiment should
inform its analysis. Fisher placed the physical act of randomization at the center
of his inferential theory, enshrining it as ``the reasoned basis'' for inference
\citep{fisher1935design}. Building on these insights,
\cite{kempthorne1955randomization} proposed a randomization theory of inference
for experiments, in which inference follows from the precise randomization
mechanism used in the design. This approach has gained popularity in the causal
inference literature because it relies on very few assumptions
\citep{neyman_1923, imbens2015causal}.

Yet the injunction to ``analyze  the way you randomize'' is not always followed in
practice, as noted by \cite{senn2004controversies} who argues that in clinical
trials the analysis does not always follow strictly from the randomization
performed. For instance, a Bernoulli randomized experiment, in which each unit is assigned to treatment or control independently of other units (for instance, by independent coin flips) leading to a random number of treated units, might be analyzed as
if it were a completely randomized experiment, in which the total number of treated units would be fixed prior to randomization.
Or, similarly, we might analyze a completely
randomized experiment as if it had been stratified.

This paper studies such as-if analyses in detail in the context of Neymanian causal inference, and makes three
contributions. First we formalize the notion of as-if analyses, motivating
their usefulness and proposing a rigorous validity criterion
(Section~\ref{sec:as_if_con_pro}). Our framework is grounded in the
randomization-based approach to inference. In the two examples we
described above, the analysis conditions on some aspect of the observed
assignment; for instance, in the first example, the complete randomization
is obtained by fixing the number of treated units to its observed value. 
The idea that inference should be conditional on quantities that affect the
precision of estimation is not new in the experimental design literature
\citep[e.g.,][]{cox1958, cox2009randomization, dawid1991fisherian} or the larger statistical
inference literature \citep[e.g.,][]{sarndal1989weighted, sundberg2003conditional},
and it has been reaffirmed recently in the causal inference literature
\citep{branson2019randomization, hennessy2016conditional}.

Our second contribution is to verify that in our setting, conditioning leads to valid
as-if analyses. We also warn against a dangerous pitfall: some as-if
analyses look conditional on the surface, but are in fact neither conditional
nor valid. This is the case, for instance, of analyzing a completely randomized
experiment by conditioning on the covariate balance
being no worse than that of the observed assignment
(Section~\ref{subsec:conditioning}).
We further point out that the motivation for the use of conditioning in analyzing experimental data should not be to increase precision, but rather to increase relevance.
Consider the case of a Bernoulli randomized experiment with 100 units and a 0.5 probability of any given unit receiving treatment.
Intuitively, if the assignment we observe has 1 treated unit and 99 control units, we should have a wider confidence interval than if the observed assignment had 50 treated units and 50 control units.
This is exactly what will occur if analyzing the experiment conditionally, as a completely randomized experiment.
The idea of relevance is discussed further in Section~\ref{subsec:relevance}.

Our third contribution is to show how
our ideas can be used to suggest new methods (Section~\ref{sec:ex}) and also show how they can be used to evaluate existing methods (Section~\ref{sec:disc_mat}).
We finish the work by discussing practical challenges (Section~\ref{append:est_des}).
Our goal in these discussions is primarily to highlight areas for future work, as this paper is primarily conceptual in nature.


\section{As-if confidence procedures}\label{sec:as_if_con_pro}

\subsection{Setup}
\label{subsec:validityset}

Consider $N$ units and let $Z_i \in \{0,1\}$ $(i=1, \ldots, N)$ be a
binary treatment indicator for unit $i$. We adopt the potential outcomes
framework \citep{rubin_1974, neyman_1923}, where, under the Stable Unit
Treatment Value Assumption \citep{rubin_1980}, each unit has two potential
outcomes, one under treatment, $Y_i(1)$, and one under control, $Y_i(0)$, and the observed response is 
$Y_i^{obs} = Z_i Y_i(1) + (1 - Z_i) Y_i(0)$. We denote by $\bZ$, $\bYone$,
and $\bYzero$ the vectors of binary treatment assignments, treatment
potential outcomes, and control potential outcomes for all $N$ units.
Let $\tau$ be our estimand of interest; in most of the examples we
take $\tau$ to be the average treatment effect
$\tau = N^{-1} \sum_{i=1}^N \{Y_i(1) - Y_i(0)\}$, but our results apply more
generally to any estimand that is a function of the potential outcome
vectors $(\bYzero, \bYone)$.
An estimator $\hat{\tau}$ is
a function of the assignment vector $\bZ$ and the observed outcomes. For
clarity, we will generally write $\hat{\tau} = \hat{\tau}(\bZ)$ to emphasize
the dependence of $\hat{\tau}$ on $\bZ$, but keep the dependence on the
potential outcomes implicit. Denote by $\eta_0$ the design describing how
treatment is allocated, so for any particular assignment $\bz$, $\eta_0(\bz)$
gives the probability of observing $\bz$ under design $\eta_0$. In
randomization-based inference, we consider the potential outcomes
$(\bYzero, \bYone)$ as fixed and partially unknown quantities; the randomness
comes exclusively from the assignment vector $\bZ$ following the distribution
$\eta_0$. The estimand $\tau$ is therefore fixed because it is a function of the
potential outcomes only, while the observed outcomes and the estimator
$\hat{\tau}(\bZ)$ are random because they are functions of the random assignment
vector $\bZ$. 

Our focus is on the construction of confidence intervals for $\tau$
under the randomization-based perspective. We define a confidence procedure as a
function $C$ mapping any assignment $\bz \in \Zset_{\eta_0}$ and associated
vector of observed outcomes, $Y^{obs}$, to an interval in $\mathds{R}$, where
$\Zset_{\eta_0} = \{\bz \in \{0,1\}^N: \eta_0(\bz) > 0\}$ is the support of the
randomization distribution. Standard confidence intervals are examples of
confidence procedures, and are usually based on an approximate distribution of
$\hat{\tau}(\bZ) - \tau$. For careful choices of $\hat{\tau}$ and $\eta_0$, the
random variable $\hat{\tau}(\bZ) - \tau$, standardized by its standard deviation
$\sqrt{\text{var}_{\etaref}(\hat{\tau})}$ induced by the design $\etaref$,
is asymptotically standard normal \citep[see][]{li2017general}. We can then
construct an interval
\begin{equation}\label{eq:non-oracle}
\left[\hat{\tau}(\bZ) - 1.96 \sqrt{\hat{V}_{\etaref}(\bZ)}, \hat{\tau}(\bZ) + 1.96 \sqrt{\hat{V}_{\etaref}(\bZ)}\right],
\end{equation}
where $\hat{V}_{\etaref}$ is an estimator of $\text{var}_{\etaref}(\hat{\tau})$.
Such a confidence interval is built from the (estimated) 2.5 and 97.5 quantiles of $\hat{\tau}(\bZ) - \tau$, shifted by $\hat{\tau}(\bZ)$.
Discussing the validity of these kinds of confidence procedures is difficult for
two reasons. First, they are generally based on asymptotics, so validity in
finite populations can only be approximate. Second, they use the square root of estimates of the
variance which, in practice, tend to be biased, especially with finite population causal inference under which $\text{var}_{\etaref}(\hat{\tau})$ may only be estimated conservatively, in generality.
These two issues are often accepted in practice but obscure the conceptual underpinnings of as-if analyses. We circumvent these issues by
focusing instead on oracle confidence procedures, which are based on the true
quantiles of the distribution of $\hat{\tau}(\bZ) - \tau$ induced by the design
$\etaref$. Specifically, we consider $1-2\alpha$ level confidence intervals
($0 < \alpha < 1/2$) of the form
\begin{equation}\label{eq:basic-oracle}
  \Cz{\etaref} = [\hat{\tau}(\bZ) - U_{\alpha}(\etaref),
  \hat{\tau}(\bZ) - L_{\alpha}(\etaref)]
\end{equation}
where $U_{\alpha}(\etaref)$ and $L_{\alpha}(\etaref)$ are the $\alpha$ upper and
lower quantiles, respectively, of the distribution of $\hat{\tau}(\bZ) - \tau$
under design $\etaref$. Because they do not depend on $\bZ$, the quantiles
$U_{\alpha}(\etaref)$ and $L_{\alpha}(\etaref)$ are fixed.
The confidence procedure in Equation~(\ref{eq:basic-oracle}) is
an oracle procedure because unlike the interval in
Equation~(\ref{eq:non-oracle}), it cannot be computed from observed data. Oracle
procedures allow us to set aside the practical issues of approximation and
estimability to focus on the essence of the problem. We discuss some of the
practical issues that occur without oracles in Section~\ref{append:est_des}.

\subsection{As-if confidence procedures}

Given data from an experiment, it is natural to consider the confidence procedure
$\Cz{\etaref}$ constructed with the design $\etaref$ that  was actually used
to randomize the treatment assignment.
This would be following the ``analyze as you randomize'' principle.
Consider, however, the oracle procedure
\begin{equation}
\Cz{\eta} = [\hat{\tau}(\bZ) - U_{\alpha}(\eta), \hat{\tau}(\bZ) - L_{\alpha}(\eta)], \label{eq:const_map}
\end{equation}
based on the distribution of $\hat{\tau}(\bZ) - \tau$ induced by some other
design $\eta$ that assigns positive probability to the observed assignment.
In this case we say that the experiment is analyzed as-if it were randomized according to
$\eta$. We generalize this idea further by allowing the design $\eta$ used
in the oracle procedure to vary depending on the observed assignment. This can
be formalized with the concept of a design map.
\begin{definition}[Design map]
\label{def:design_mapping}
  Let $\mathcal{D}$ be the set of designs (probability distributions) with support in $\{0,1\}^N$.
  A
  function $H: \Z_{\etaref} \rightarrow \mathcal{D}$ which maps each
  assignment $\bz \in \Zset_{\eta_0}$ to a design $H(\bz) \in \mathcal{D}$
  is a design map.
\end{definition}
As an example, the design of the experiment, $\eta_0$, is one element of the set of designs $\mathcal{D}$ (i.e., $\eta_0 \in \mathcal{D}$).

A confidence procedure $\Cz{H(\bZ)}$ can then be constructed using
design map $H$ as follows:
\begin{equation*}
\Cz{H(\bZ)}= [\hat{\tau}(\bZ) - \Uz{H(\bZ)}, \hat{\tau}(\bZ)-\Lz{H(\bZ)}].
\end{equation*}
This is an instance of as-if analysis, in which the design used to
analyze the data depends on the observed assignment. That is, while
we traditionally have one rule for how to create a confidence
interval, we may now have many rules, possibly as many as
$|\Z_{\etaref}|$, specified via the design map. In the special
case in which the design map $H$ is constant, i.e $H(\bZ) = \eta$
for all $\bZ \in \Z_{\etaref}$, we write $\Cz{\eta}$ instead of $\Cz{H(\bZ)}$,
with a slight abuse of notation, leading to Equation~\ref{eq:const_map}.
Note that the design map function itself is fixed before observing the treatment assignment.
\setcounter{ex1}{\value{example}}
\begin{example}
\label{example:bern-crd}
Consider an experiment run according to a Bernoulli design with probability of
treatment $\pi = 0.5$, where, to ensure that our estimator is defined, we remove assignments with no treated or control
units. That is, $\Z_{\eta_0}$ is the set of all assignments such that at least
one unit receives treatment and one unit receives control, and
$\etaref = \text{Unif}(\Z_{\etaref})$. Let $\eta_k$ ($k = 1, \ldots, N-1$) be
the completely randomized design with $k$ treated units. We use the design map
$H(\bZ) = \eta_{N_1(\bZ)}$ where $N_1(\bZ) = \sum_{i=1}^N Z_i$. That is, we analyze
the Bernoulli design as if it were completely randomized, with the observed
number of treated units assigned to treatment. Consider the concrete case where
$N = 10$. Suppose that we observe assignment $\bz$ with $N_1(\bz) = 3$. The
design $H(\bz) = \eta_3$ corresponds to complete randomization with 3
units treated out of 10, and the confidence procedure
$C(\bz;H(\bz))$ is constructed using the distribution of
$\hat{\tau}(\bZ)-\tau$ induced by $\eta_3$. We analyze as if a completely
randomized design with 3 treated units had actually been run. If instead we
observe $\bz^*$ with $N_1(\bz^*) = 4$, then the confidence procedure
$C(\bz^*;H(\bz^*))$ would be constructed by considering the distribution
of $\hat{\tau}(\bZ)-\tau$ induced by $\eta_4$.
\end{example}

\setcounter{ex2}{\value{example}}
\begin{example}
\label{example:block}
Let our units have a categorical covariate $X_i$.
Categorical covariates form blocks: they partition the sample based on covariate value.
Assume that the actual experiment was run using complete randomization with fixed number of treated units $N_1$, but discarding assignments where there is not at least one treated unit and one control unit in each block.
That is, $\Z_{\eta_0}$ is the set of all assignments with $N_1$ treated units such that at least one unit receives treatment and one unit receives control in each block and $\etaref = \text{Unif}(\Z_{\etaref})$.
This restriction on the assignment space is to account for the associated blocked estimator being undefined; with moderate size blocks we can ignore this nuisance event due to its low probability.
For vector $\kappa$ whose $j$th entry is an integer strictly less than the size of the $j$th block and strictly greater than 0, let $\eta_{\kappa}$ be a block randomized design with the number of treated units in each block corresponding to the numbers given in vector $\kappa$.
We use the design map $H(\bZ) = \eta_{N_{1,block}(\bZ)}$ where $N_{1,block}(\bZ)$ is the vector that gives the number of treated units within each block and $\eta_{N_{1,block}(\bZ)}$ is the blocked design corresponding to $N_{1,block}(\bZ)$.
We use post-stratification\footnote{The term post-stratification originates in the survey sampling literature but has been brought into the design literature to describe the act of analyzing an experiment as a blocked experiment post randomization \citep[see][]{Miratrix2013}.} \citep{holt1979post, Miratrix2013} to analyze this completely randomized design as if it were block randomized.
\end{example}

\setcounter{exnew}{\value{example}}
\begin{example}
\label{example:fac}
Consider an experiment with factorial structure in the treatments, particularly with two factors of interest each with two levels, active and passive (the case with more factors or levels is immediate).
Assume that the actual experiment was run using complete randomization on each factor separately, such that there were exactly $N_1$ units receiving the active level for each factor.
This marginal randomization means that the number of units assigned to the combined levels of both factors is random.
To ensure our estimator is well defined, we assume that there is a non-zero number of units assigned to each treatment combination.
Define $n_{z_1,z_2}$ as the observed number of units  assigned to level $z_1$ of factor 1 and level $z_2$ of factor 2.
Then $\Z_{\eta_0}$ is the set of all assignments such that $0<N_1 < N$ units receive the active level of factor 1 and  $0<N_1 < N$ units receive the active level of factor 2, and $n_{z_1,z_2}>0$ for all $z_1$ and $z_2$.
Let $\etaref = \text{Unif}(\Z_{\etaref})$.
For vector $\gamma$ of length four whose $j$th entry is an integer strictly less than $N$ and strictly greater than 0 such that $\sum_j\gamma_j = N$, let $\eta_{\gamma}$ be a standard factorial design with complete randomization on the joint treatment groups, with treatment group sizes corresponding to the numbers given in vector $\gamma$.
We use the design map $H(\bZ) = \eta_{n_{Z_1,Z_2}}$, where $n_{Z_1,Z_2}$ is the vector of the four observed $n_{z_1,z_2}$.
$H(\bZ)$ is a factorial design where the vector $n_{Z_1,Z_2}$ gives the number of units assigned to each treatment combination.
Then this design map leads to the analysis as a factorial design with complete randomization, fixing the number of units in each treatment group (rather than just the marginal number for each factor separately).
\end{example}

\setcounter{ex3}{\value{example}}
\begin{example}
\label{example:rerand}
Assume that the actual experiment was run using complete randomization with exactly $N_1$ units treated.
That is, $\Z_{\eta_0}$ is the set of all assignments such that $0<N_1 < N$ units receive treatment and $\etaref = \text{Unif}(\Z_{\etaref})$.
Let $X_i$ be a continuous covariate for each unit $i$. 
Define a covariate balance measure $\Delta_X(\bZ)$, e.g.
\begin{equation*}
		\Delta_X(\bZ) 
        = \frac{1}{N_1} \sum_i Z_i X_i 
        - \frac{1}{N-N_1} \sum_i (1-Z_i) X_i.
\end{equation*}
For any assignment $\bz$, define $\mathcal{A}(\bz) = \{\bZ: |\Delta_X(\bZ)| \leq | \Delta_X(\bz)|\}$, which will give us the set of assignments with covariate balance better than or equal to the observed covariate balance.
We use the design map $H: \bz \rightarrow \text{pr}_{\etaref}\{\bZ \mid \bZ \in \mathcal{A}(\bz)\}$.
Then this design map leads to analyzing the completely randomized design as if it were rerandomized \citep[see][for more on rerandomization]{morgan2012rerandomization} with acceptable covariate balance cut-off equal to that of the observed covariate balance.
\end{example}

\setcounter{ex4}{\value{example}}
\begin{example}
\label{example:crd}
Assume that the actual experiment was run using block randomization where $N_{1,block}$ is a fixed vector, and hence we can drop the $\bZ$ in the notation used in Example~\ref{example:block}, that gives the number of treated units within each block.
That is, $\Z_{\eta_0}$ is the set of all assignments such that the number of treated units in each block is given by $N_{1,block}$ and $\etaref = \text{Unif}(\Z_{\etaref})$.
Let $\eta$ correspond to a completely randomized design, as laid out in Example~\ref{example:block}, with $N_1$, the total number of units treated across all blocks, treated.
We use the constant design map $H(\bZ) = \eta$.
This corresponds to analyzing this block randomized design as if it were completely randomized.
\setcounter{ex5}{\value{example}}
\end{example}
Throughout the paper, we focus on settings in which the same
estimator is used in the original analysis and in the as-if analysis. In
practice, the two analyses might employ different estimators. For instance,
in Example~\ref{example:block}, we might analyze the completely randomized
experiment with a difference-in-means estimator, but use the standard blocking
estimator to analyze the as-if stratified experiment. We discuss this point
further in Section~\ref{append:est_des} but in the rest
of this article, we fix the estimator and focus on the impact of changing
only the design.

\subsection{Validity, relevance and conditioning}
\label{subsec:relevance}

We have formalized the concept of an as-if analysis, but we have not yet 
addressed an important question: why should we even consider such analyses
instead of simply analyzing the way we randomize? Before we answer this
question, we first introduce a minimum validity criterion for as-if
procedures.
\begin{definition}[Valid confidence procedure]
  Fix $\gamma \in [0,1]$. Let $\etaref \in \mathcal{D}$ be the design used
  in the original experiment and let $H$ be a design map on $\Z_{\etaref}$.
  The confidence procedure $\Cz{H(\bZ)}$ is said to be valid with respect
  to $\etaref$, or $\etaref$-valid, at level $\gamma$ if
  $\text{pr}_{\etaref}\{\tau \in \Cz{H(\bZ)}\} \geq \gamma$. When a procedure is
  valid at all levels, we simply say that it is $\etaref$-valid. 
\end{definition}
This criterion is intuitive: a confidence procedure is valid if its coverage is
as advertised over the original design. The following simple result formalizes
the popular injunction to ``analyze the way you randomize''
\citep[][p. 317]{lachin1988properties}:
\begin{proposition}\label{prop:nulldesignmap}
    The procedure $\Cz{\etaref}$ is $\etaref$-valid.
\end{proposition}
Given that the procedure $\Cz{\etaref}$ is $\etaref$-valid, why should we
consider alternative as-if analyses, even valid ones? That is, having observed
$\bZ$, why should we use a design $H(\bZ)$ to perform the analysis, instead of
the original $\etaref$?
%
%
A natural, but only partially correct, answer would be that the goal of an
as-if analysis is to increase the precision of our estimator and obtain
smaller confidence intervals while maintaining validity. After all, this is the
purpose of restricted randomization approaches when considered at the design
stage. For instance, if we have reasons to believe that a certain factor might
affect the responses of experimental units, stratifying on this factor will
reduce the variance of our estimator. This analogy, however, is misleading. The primary goal of an as-if analysis is not to increase the precision of
the analysis but to increase its relevance. 
In fact, we argue heuristically in Section~\ref{append:est_des} that (given a conditionally unbiased estimator) an as-if analysis will not increase precision on average, over the original assignment distribution.
Rather, it is frequently the change of estimator that has created this impression.
 
Informally, an observable quantity is relevant if a reasonable person cannot tell if a confidence interval will be too long or too short given that quantity.
The concept of relevance captures the
idea that our inference should be an accurate reflection of our uncertainty
given the observed information from the realized randomization; our confidence intervals should be narrower
if our uncertainty is lower and wider if our uncertainty is higher.
Defining the concept of relevance formally is difficult.
See \cite{buehler1959some} and \cite{robinson1979conditional} for a more formal treatment.
Our appendix~\ref{append:simp_rel} gives a precise discussion in the context of betting games, following \cite{buehler1959some}.
We will not attempt a formal definition here and instead, following \cite{liu2016there}, we will illustrate its essence with a simple example.

Consider the Bernoulli design scenario of Example~\ref{example:bern-crd} with 100 units. From
Equation~\eqref{eq:basic-oracle}, the oracle interval $\Cz{\etaref}$ constructed
from the original design has the same length regardless of the assignment vector
$\bZ$ actually observed. Yet, intuitively, an observed assignment vector with
$1$ treated unit and $99$ control units should lead to less precise inference,
and therefore wider confidence intervals than a balanced assignment vector with
$50$ treated units and $50$ control units. In a sense, the confidence interval
$\Cz{\etaref}$ is too narrow if the observed assignment is severely imbalanced,
too wide if it is well balanced, but right overall.
Let $\Zset_1$ be the set of all assignments with a single treated unit
and $\Zset_{50}$ the set of all assignments with 50 treated units. If the
confidence interval $\Cz{\etaref}$ has level $\gamma$, we expect
\begin{align*}
  \text{pr}_{\etaref}\{\tau \in \Cz{\etaref} \mid \bZ \in \Zset_1\} \leq \gamma; \quad 
  & \text{pr}_{\etaref}\{\tau \in \Cz{\etaref} \mid \bZ \in \Zset_{50}\} \geq \gamma;
\end{align*}
where the inequalities are usually strict. See appendix~\ref{append:simp_rel} for a proof in a concrete setting. More formally,
we say that the procedure is valid marginally, but is not valid conditional on the
number of treated units. This example is illustrated in
Figure~\ref{fig:bern_cond}, which shows that the conditional coverage is below 0.95 if the
proportion of treated units is not close to 0.5, and above 0.95 if the
proportion is around 0.5. To remedy this, we should use wider confidence
intervals in the first case, and narrower ones in the second. The right panel of
Figure~\ref{fig:bern_cond} shows that in this case, a large fraction of assignments
have a proportion of treated units close to 0.5, therefore the confidence interval
$\Cz{\etaref}$ will be too large for many realizations of the design $\eta_0$.
\begin{figure}
	\centering
	\includegraphics[scale=0.53]{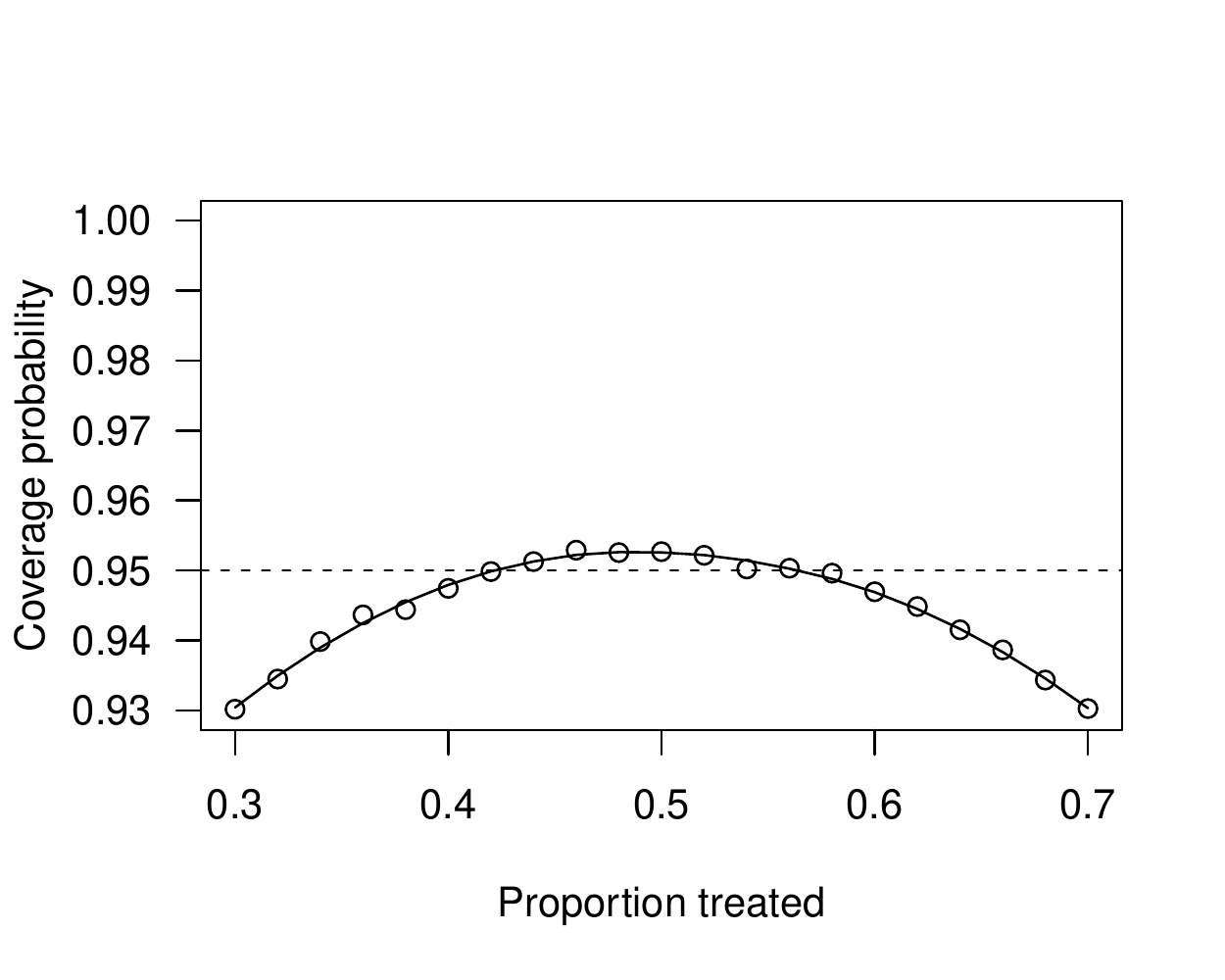}
	\includegraphics[scale=0.53]{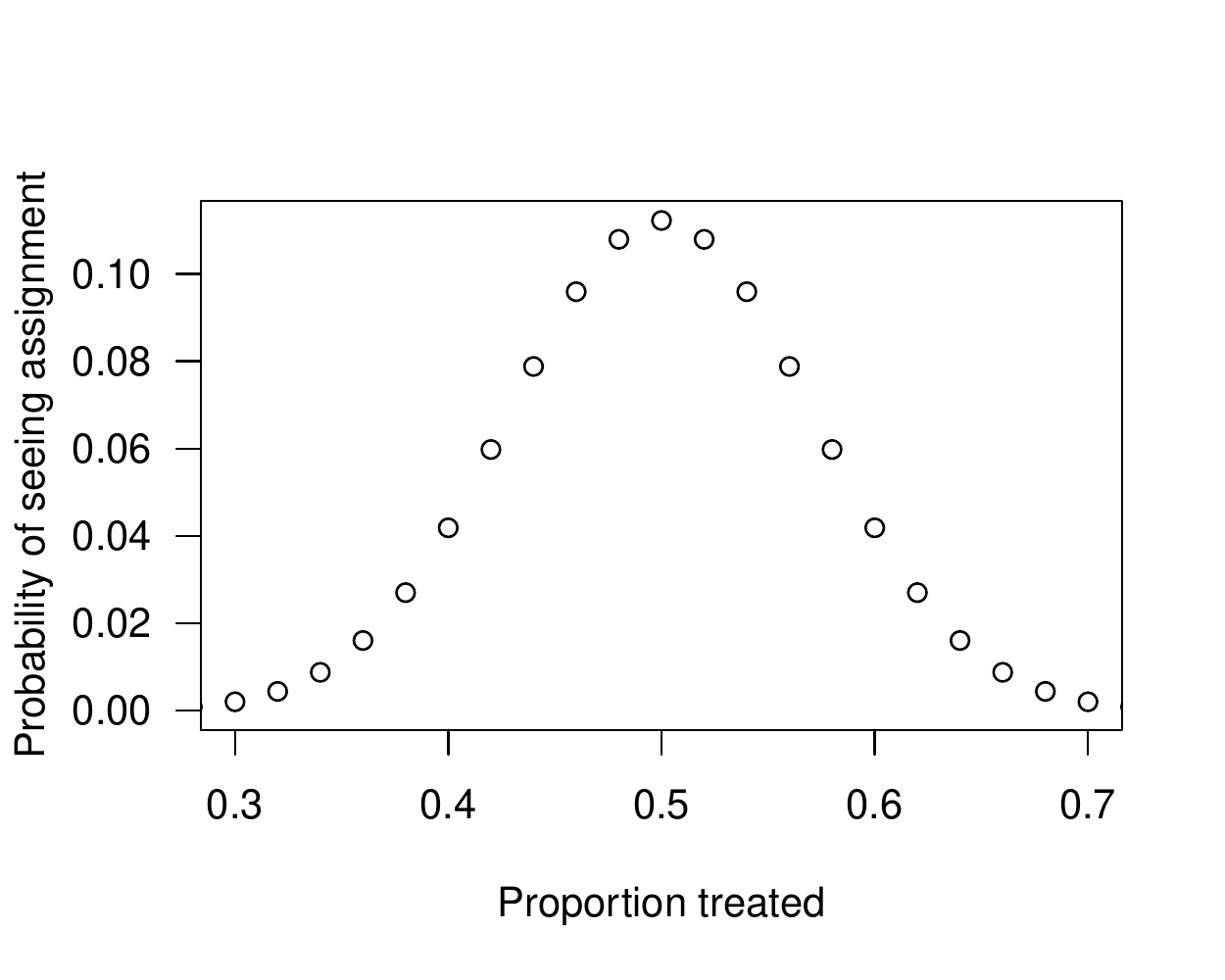}
        \caption{Left: conditional coverage for a Bernoulli experiment with 100 units
          each having probability 0.5 of being treated. Right: distribution of the
          proportion of treated units for this Bernoulli experiment.
          }
    \label{fig:bern_cond}
\end{figure}
Our confidence interval should be relevant to the assignment vector
actually observed, and reflect the appropriate level of uncertainty. In the
context of randomization-based inference, this takes the form of an as-if
analysis.


The concept of relevance and its connection to conditioning has a long history
in statistics.~\cite{cox1958} gives a dramatic example of a scenario in which
two measuring instruments have widely different precisions, and one of them is
chosen at random to measure a quantity of interest.
Cox argues that the
relevant measure of uncertainty is that of the instrument actually used, not
that obtained by marginalizing over the choice of the instrument.
In other words,
the analysis should be conditional on the instrument actually chosen. This is an
illustration of the conditionality principle \citep{birnbaum1962foundations}. 
Many of the examples we gave earlier are similar in spirit to this example.
In the context of randomization-based inference, this conditioning argument leads to
valid as-if analyses, as we show in Section~\ref{subsec:conditioning}. An
important complication, explored in Section~\ref{subsec:conditioning}, is that
conditional as-if analyses are only a subset of possible as-if analyses, and
while the former are guaranteed to be valid, the latter enjoy no such guarantees.

\section{Conditional as-if analyses}
\label{subsec:conditioning}

\subsection{Conditional design maps}\label{subsec:cond_des_map}

We define a conditional as-if analysis as an analysis conducted with a conditional design map as defined below.
\begin{definition}\label{def:cond-design-map}[Conditional design map]
  Consider an experiment with design $\etaref$. Take any function
  $w: \Zset_{\etaref} \rightarrow \Omega$, for some set $\Omega$, and for $\upsilon \in \Omega$ define the design
  $\eta_{\upsilon} \in \mathcal{D}$ as
  $\eta_\upsilon(\bz) = \text{pr}_{\eta_0}\{\bz \mid w(\bz) = \upsilon\}$.
  That is, $\eta_{\upsilon}$ is the design which conditions on $w(\bZ) = \upsilon$.
  Then a function
  $H: \Zset_{\etaref} \rightarrow \mathcal{D}$ where
  $H(\bZ) = \eta_{w(\bZ)}$ is called a conditional design map (where here we condition on the observed value $w(\bZ)$).
\end{definition}
It is easy to verify that a conditional design map also satisfies
Definition~\ref{def:design_mapping}.
For $\bz \in \Zset_{\etaref}$, $H(z)$ is a design, not the probability of $\bz$ under some design.
The probability of any assignment $\bz' \in \Zset_{\etaref}$ under design $H(\bz)$ is
$H(\bz)(\bz') = \text{pr}_{\eta_0}\{\bZ =\bz' \mid w(\bZ) = w(\bz)\}$, where the probability is that induced by $\eta_0$. We introduced the shorthand
$H(\bZ)(\bz) = \eta_{w(\bZ)}(\bz)$ in order to ease the notation.  

For an alternative perspective on conditional design maps, notice that any
function $w: \Zset_{\etaref} \rightarrow \Omega$ induces a partition 
$\mathcal{P}_w = \{\Zset_{\upsilon}\}_{\upsilon \in \Omega}$ of the support
$\Zset_{\etaref}$, where
$\Zset_{\upsilon} = \{\bz \in \Zset_{\etaref}: w(\bz) = \upsilon\}$.
The corresponding conditional design map would then, for a given $\bZ$, restrict and renormalize the original $\etaref$ to the $\Zset_{\upsilon}$ containing $\bZ$.
An important note is that the mapping function $w$, and therefore the partitioning of the assignment space, must be fixed before observing the treatment assignment.

\setcounter{example}{\value{ex1}}
\begin{example}[cont.]
  The design map $H(\bZ)$ in Example~\ref{example:bern-crd} is a conditional
  design map, with $w(\bZ) = N_1(\bZ)$. Here we partition the assignments by the
  number of treated units.
\end{example}
\setcounter{example}{\value{ex2}}
\begin{example}[cont.]
  The design map $H(\bZ)$ in Example~\ref{example:block} is a conditional design
  map, with $w(\bZ) = N_{1,block}(\bZ)$. Here we partition the assignments by
  the vector of the number of treated units in each block.
\end{example}
While Definition~\ref{def:cond-design-map} implies
Definition~\ref{def:design_mapping}, the converse is not true: some design maps
are not conditional. For instance, the design maps we consider in
Example~\ref{example:rerand} and Example~\ref{example:crd} are not conditional, as will be discussed in
Section~\ref{subsec:nonconditional}. We can now state our main validity result.

\begin{theorem}\label{corr:main_cond}
  Consider a design $\etaref$ and a function
  $w: \Zset_{\etaref} \rightarrow \Omega$.
  Then an oracle procedure built with the conditional design map
  $H(\bZ) = \eta_{w(\bZ)}$ is $\etaref$-valid.
\end{theorem}
Proof of Theorem~\ref{corr:main_cond} is provided in appendix~\ref{append:sec_proof}. In fact, the intervals obtained are not
just valid marginally; they are also conditionally valid within each
$\Zset_{\upsilon} \in \mathcal{P}_w$, in the sense that
\begin{equation*}
  \text{pr}_{\etaref}\{\tau \in \Cz{H(\bZ)}\mid w(\bZ)=\upsilon\} = \gamma
\end{equation*}
for any $\upsilon \in \Omega$.
Conditional validity implies unconditional validity because if we have valid inference for each partition of the assignment space then we will have validity over all partitions; similar arguments regarding conditional inference have been made previously \citep[for example,][states unconditional coverage of conditional intervals with correct conditional coverage as trivial on page 84]{dawid1991fisherian}.
However, this result is novel in the randomization-based causal inference framework, to our knowledge.
Conditional validity is good: it implies increased relevance, at least with respect to
function $w$. We discuss this connection more in appendix~\ref{append:simp_rel} in the context of betting games.
\begin{corollary}\label{example:partition-2}
Let $\mathcal{P} = \{\Z_{(1)}, \ldots, \Z_{(K)}\}$ be a partition of the set of all possible assignments, $\Z_{\etaref}$ and
	\begin{equation*}
		w(\bZ) =\sum_k^K k \iv\{\bZ \in \Z_{(k)}\}.
	\end{equation*}
	Then $w: \Z_{\etaref} \rightarrow \{1,\dots,K\}$ indexes the partition that $\bZ$ is in. Using $H(\bZ) = \eta_{w(\bZ)}$ gives an $\etaref$-valid procedure, as a consequence of Theorem~\ref{corr:main_cond}.
\end{corollary}
Details are provided in appendix~\ref{append:sec_proof}.
Corollary~\ref{example:partition-2} states that any partition $\mathcal{P}$
of the support $\Zset_{\eta_0}$ induces a valid oracle confidence procedure;
having observed assignment $\bZ$, one simply
needs to identify the unique element $\Zset \in \mathcal{P}$ containing
$\bZ$ and construct an oracle interval using the design
obtained by restricting $\eta_0$ to the set $\Zset$.

An additional benefit of using conditional design maps is replicability.
Consider Example~\ref{example:bern-crd}, and the corresponding discussion of relevance with Bernoulli designs in Section~\ref{subsec:relevance}.
Under the original analysis for the Bernoulli design we would expect that the estimates for the bad randomizations with an extreme proportion of treated units will be far from the truth, on average.
But if we do not adjust the estimated precision of our estimators based on this information, we may not only have an estimate that is far from the truth but our confidence intervals will imply confidence in that poor estimate.
Although our conditional analysis will cover the truth the same proportion of the time as the original analysis, we would expect the length of our confidence intervals to reflect less certainty when we have a poor randomization.
In terms of replicability, this means that we are less likely to end up being confident in an extreme result.


\subsection{Non-conditional design maps}
\label{subsec:nonconditional}
Theorem~\ref{corr:main_cond} states that a sufficient condition for an
as-if procedure to be valid is that it be a conditional as-if procedure.
Although this condition is not necessary, we will now show that some
non-conditional as-if analyses can have arbitrarily poor properties.
Example~\ref{example:rerand}, in particular, provides a sharp illustration of this phenomenon
and, although it is an edge case, it helps build intuition for why some design
maps are not valid.

\setcounter{example}{\value{ex3}}
\begin{example}[cont.]
  The design map $H(\bZ) = \text{pr}_{\etaref}\{\bz \mid \bz \in \mathcal{A}(\bZ)\}$
  introduced in Example~\ref{example:rerand} is not a conditional design
  map. This can be seen by noticing that the sets
  $\{\mathcal{A}(\bZ)\}_{\bZ\in \Zset_{\etaref}}$ where
  $\mathcal{A}(\bZ) = \{\bz: |\Delta_X(\bz)| \leq | \Delta_X(\bZ)|\}$ do not
  form a partition of $\Zset_{\etaref}$.
\end{example}
This example is particularly deceptive because the design map $H(\bZ)$ does
involve a conditional distribution. And yet, it is not a conditional design
map in the sense of Definition~\ref{def:cond-design-map} because it does not
partition the space of assignments; each assignment $\bz$, except for the
assignments with the very worst balance, will belong to multiple $\mathcal{A}(\bZ)$.
Therefore Theorem~\ref{corr:main_cond} does not apply.

Consider the special case where covariates $X_i$ are drawn from a continuous distribution and $Y_i(0) = Y_i(1) = X_i$ ($i=1, \ldots, N$)\footnote{If $X_i$ were instead low-dimensional and discrete, we could easily condition on covariate balance. This would be equivalent to post-stratification. But with many covariates or continuous covariates, post-stratification can be infeasible.}.
We are interested in the average treatment effect, which is the difference in mean potential outcomes under treatment versus control.
Suppose that assignments are balanced such that half the units are assigned to treatment and half are assigned to control.
Then given any $\bz \in \Zset_{\eta_0}$, with probability one there are only two assignments with exactly the same value for $|\Delta_X(\bz)|$, $\bz$ and the assignment $1-\bz$; see appendix~\ref{append:sec_proof} for proof of this statement.
By construction, then, our assignment is one of two worst case assignments in terms of balance for the set $\mathcal{A}(\bZ)$.
Under the model $Y_i(0) = Y_i(1) = X_i$, the observed difference, $\hat{\tau}(\bZ) - \tau = \hat{\tau}(\bZ)$ will be the most extreme in $\mathcal{A}(\bZ)$ and thus $\tau$ would lie outside the oracle confidence interval if, as is typical, the set is large enough with $|\mathcal{A}(\bZ)| \geq 2/\gamma$, where $|\mathcal{A}(\bZ)|$ is the size of set $\mathcal{A}(\bZ)$.
Thus this design map would lead to poor coverage.
In fact, we show in appendix~\ref{append:sec_proof} that if we instead make the inequality strict and take $\mathcal{A}(\bZ) = \{\bz: |\Delta_X(\bz)| < | \Delta_X(\bZ)|\}$, the as-if procedure of Example~\ref{example:rerand} has a coverage of $0$.
Intuitively, this is because the observed assignment $\bZ$ always has the worst covariate balance of all assignments within the support $\mathcal{A}(\bZ)$. 
Although extreme, this example illustrates the fact that as-if analyses are not guaranteed to be valid if they are not conditional. 

\setcounter{example}{\value{ex4}}
\begin{example}[cont.]
  The design map
  introduced in Example~\ref{example:crd}, in which we analyze a blocked design as if it were completely randomized, is also not a conditional design
  map. This can be seen by noticing that the complete randomization does not partition the blocked design but rather the blocked design is a subset, or a single element of a partition, of the completely randomized design.
\end{example}
This implies that Example~\ref{example:crd} can also lead to invalid analyses; if the blocked design originally used is a particularly bad partition of the completely randomized design, in the sense of having wider conditional intervals, we will not have guaranteed validity using a completely randomized design for analysis.
See \cite{pashley2020bkcr} for further discussion on when a blocked design can result in higher variance of estimators than an unblocked design. 
This is a tricky case as typically analyzing a blocked experiment as completely randomized will lead to conservative estimation (due to inflated estimation of variance), but such a result is not guaranteed.

\subsection{How to build a better conditional analysis}

The original goal of the as-if analysis of Example~\ref{example:rerand} was
to incorporate the observed covariate balance in the analysis to increase
relevance. We have shown that the design map originally proposed was not a
conditional design map. We now show how to construct a conditional design map,
and therefore a valid procedure, for this problem. The idea is to partition the
support $\Zset_{\etaref}$ into sets of assignments with similar covariate
balance and then use the induced conditional design map, as prescribed by
Corollary~\ref{example:partition-2}.
Let $\{\Delta_X(\bz): \bz \in \Z_{\etaref}\}$ be the set of all possible
covariate imbalance values achievable by the design $\etaref$, and 
$\mathcal{G} = \{\mathcal{G}_{(1)}, \ldots, \mathcal{G}_{(K)}\}$ be a partition
of that set into $K$ ordered elements. That is, for
any $k, k'$ with $k<k'$, we have $\delta < \delta'$  for all
$\delta \in \mathcal{G}_{(k)}, \, \delta' \in \mathcal{G}_{(k')}$.
This induces a partition 
$\mathcal{P} = \{\Z_{(1)}, \ldots, \Z_{(K)}\}$ of $\Z_{\etaref}$, where
\begin{equation*}
	\Z_{(k)} = \{\bz \in \Zset_{\etaref}: \Delta_X(\bz) \in \mathcal{G}_{(k)}\}.
\end{equation*}
Now we can directly apply the results of Corollary~\ref{example:partition-2}. This
approach is similar in spirit to the conditional randomization tests proposed
by \cite{rosenbaum1984conditional}; see also \cite{branson2019randomization} and
\cite{hennessy2016conditional}. The resulting as-if analysis improves on the
original analysis under $\etaref$ by increasing its relevance. Indeed, suppose
that the observed assignment has covariate balance $\delta$. Then the confidence
interval constructed using $\etaref$ will involve all of the assignments in
$\Zset_{\etaref}$, including some whose covariate balances differ sharply from
$\delta$. In contrast, the procedure we just introduced restricts the randomization
distribution to a subset of assignments $\Zset_{(k)}$ containing only assignments
with balance close to $\delta$.

This does not, however, completely solve the original problem. Suppose, for
instance, that by chance, $\Delta_{X}(\bZ) = \max \mathcal{G}_{(k)}$. By
definition, the randomization distribution of the as-if analyses we introduced
above will include the assignment $\bz$ such that
$\Delta_{X}(\bz) = \min \mathcal{G}_{(k)}$, but not $\bz^*$ such that
$\Delta_{X}(\bz^*) = \min \mathcal{G}_{(k+1)}$ even though $\bz^*$ might be more
relevant to $\bZ$ than $\bz$, in the sense that we may have
$|\Delta_{X}(\bZ) - \Delta_{X}(\bz^*)| \leq |\Delta_{X}(\bZ) - \Delta_{X}(\bz)|$.
This issue does not affect validity, but it raises concerns about relevance when the
observed assignment is close to the boundary of a set $\Zset_{(k)}$. 
Informally, we would like to choose $\Zset_{(k)}$ in such a way that 
the observed assignment $\bZ$ is at the center of the set, as measured by
covariate balance. For instance, fixing $c > 0$, we would like to construct an
as-if procedure that randomizes within a set of the form
$\mathcal{B}(\bZ) = \{\bz: \Delta_X(\bz) \in
[\Delta_X(\bZ)-c, \Delta_X(\bZ)+c]\}$, rather than $\Zset_{(k)}$. A naive
approach would be to use the design mapping
$H: \bZ \rightarrow \text{pr}_{\etaref}\{\bz \mid \bz \in \mathcal{B}(\bZ)\}$,
but this is not a conditional design mapping. \cite{branson2019randomization}
discussed a similar approach in the context of randomization tests and also noted
that it was not guaranteed to be valid.

Let's explore further why $H: \bZ \rightarrow \text{pr}_{\etaref}\{\bz \mid \bz \in \mathcal{B}(\bZ)\}$ does not have guaranteed validity.
In this case, each assignment vector has an interval or window of acceptable covariate balances centered around it.
The confidence interval for a given $\bZ \in \Z_{\etaref}$ is guaranteed to have $\gamma \cdot 100$\% coverage over all assignments within the window of covariate balances defined by set $\mathcal{B}(\bZ)$.
So, if we built a confidence interval for each assignment in $\mathcal{B}(\bZ)$, using the design conditioning on $\mathcal{B}(\bZ)$, $\gamma \cdot 100$\% of those intervals would cover the truth.
    However, we would only ever observe these intervals for those assignments with exactly the same covariate balance as $\bZ$; other assignments would get mapped to different $B(\bZ')$.
    Furthermore, there are no guarantees about which assignments will result in a confidence interval covering the truth. 
    Over the smaller subset of assignments with exactly the same covariate balance as $\bZ$, which lead to the same design over $\mathcal{B}(\bZ)$, the coverage may be less than $\gamma \cdot 100$\%.

To build a solution with guaranteed validity, we need more flexible tools.
The following section will discuss how we can be more flexible, while still guaranteeing validity, by introducing some randomness.


\section{Stochastic conditional as-if}
\label{sec:ex}
\subsection{Stochastic design maps}\label{sec:stoch_dm}
The setting of Example~\ref{example:rerand} posed a problem of how to build valid procedures that allow the design mapping to vary based on the assignment.
That is, we want to avoid making a strict partition of the assignment space but still guarantee validity.
We can do this by introducing some randomness into our design map.
\begin{definition}\label{def:stochastic_mapping}[Stochastic conditional design map]
  Consider an experiment with design $\etaref$. For observed assignment
  $\bZ \in \Z_{\etaref}$, draw an additional bit of randomness $w \sim m(w \mid \bZ)$ from a given distribution
  $m( \cdot \mid \bZ)$, indexed by $\bZ$ and with support on some set $\Omega$, and consider the design
  \begin{equation*}
    \eta_{w}(\bz) = \text{pr}_{\etaref}\{\bz \mid w\} \propto m(w \mid \bz) \text{pr}_{\etaref}\{\bz\}.
  \end{equation*}
  The mapping $H: \bZ \rightarrow \mathcal{D}$, with $H(\bZ) = \eta_{w}$ and $w \sim m(w \mid \bZ)$, is
  called a stochastic design mapping.
\end{definition}
This $w$ is our bit of randomness that will allow us to blend our conditional maps to regain validity.
In the special case where the distribution $m(w \mid \bZ)$ degenerates into $\delta_{w = w(\bZ)}= \iv(w = w(\bZ))$,
Definition~\ref{def:stochastic_mapping} is equivalent to
Definition~\ref{def:cond-design-map}. When $m$ is non-degenerate, the
stochastic design map $H$ becomes a random function.

Before stating our theoretical result for stochastic design maps, we first
examine how the added flexibility that these maps afford can be put to use in the
context of Example~\ref{example:rerand}. 
Let $c > 0$, $\Omega = \mathbb{R}$, and define
\begin{equation*}
  m(w \mid \bZ)
  = \text{Unif}\bigg(w' : |\Delta_X(\bZ) - w'| \leq c\bigg)
  = \text{Unif}\bigg(\Delta_X(\bZ) - c, \Delta_X(\bZ) + c\bigg).
\end{equation*}
Our $m(w \mid \bZ)$ selects a $w^{obs}$ near the observed imbalance.
Having observed $\bz \in \Zset_{\etaref}$ and drawn $w^{obs} \sim m(w \mid \bZ=\bz)$, we then
consider the design
\begin{equation*}
  H(\bZ) = \text{pr}_{\etaref}\{\bz \mid w^{obs}\} =
  \begin{cases}
    \frac{\text{pr}_{\etaref}\{\bz\}}{\nu(w^{obs})} & \text{ if } \Delta_X(\bz) \in [w^{obs} - c, w^{obs} + c] \\
    0 & \text{ otherwise, } 
  \end{cases}
\end{equation*}
with normalizing factor $\nu(w) = \text{pr}_{\etaref}\{w\} / m(w \mid \bZ) = 2c \text{pr}_{\etaref}\{w\}$ (with some adjustment due to truncation at the extremes). In other
words, we analyze the experiment by restricting the randomization to
a set $\mathcal{A}(w) = \{\bz: \Delta_X(\bz) \in [w-c, w+c]\}$.
Comparing $\mathcal{A}(w)$ to our original randomization set $\mathcal{B}(\bZ)$,
we see that while $\mathcal{B}(\bZ)$ is a set of $w$ imbalances centered on observed $\Delta_X(\bZ)$, $\mathcal{A}(w)$ is only centered on $\Delta_X(\bZ)$ on average over draws of $w$.
The following theorem guarantees that this stochastic procedure is valid. The proof is in appendix~\ref{append:sec_proof}.
\begin{theorem}\label{theorem:main_cond}
  Consider a design $\etaref$ and a variable $w$, with conditional distribution
  $m(w \mid \bZ)$. Then an oracle procedure built at level $\gamma$ with the stochastic conditional design map $H(\bZ)$, which draws $w^{obs}$ and maps to
  $\eta_{w^{obs}} = \text{pr}_{\etaref}\{\bz \mid w^{obs}\}$, is $\etaref$-valid at level
  $\gamma$.
\end{theorem}
Stochastic conditional design maps mirror conditioning mechanisms introduced by \cite{basse2018exact} in the context of randomization tests. 
Inference is also stochastic here in the sense that a single draw of $w$ determines the final reference distribution to calculate the confidence interval.
We discuss next discuss a way to aggregate results from multiple draws of $w$.

\subsection{Aggregating random confidence intervals}\label{append:fuzzy}
The randomness intrinsic to the method introduced in this section is similar to the introduction of randomness
into uniformly most powerful (UMP) tests \citep[see][]{Lehmann2005}. Instead of having one's
inference depend on a single draw of $w$, one can use fuzzy confidence intervals
\citep[see][]{geyer2005fuzzy} to marginalize over the distribution $m$.
In a fuzzy interval, similar to fuzzy sets, membership in the interval is not binary but rather allowed to take values in $[0,1]$.

Consider an oracle confidence procedure $C(\hat{\tau}(\bZ); H(Z))$ built with a
stochastic conditional design map $H$.
The observed confidence interval for a fixed, observed $\bz$, $C(\bz) = C(\hat{\tau}(\bz); H(\bz))$, is still a random variable because $H$ is stochastic.
Specifically, $C(\bz) = C(\bz,w)$ depends on the draw $w \sim m(w \mid \bz)$, so for
a given observed dataset, we obtain a distribution $P(C(\bz,w))$ of confidence intervals,
where the randomness is that induced by $w \sim m(w \mid \bz)$. One useful way to summarize
this information is to construct a fuzzy confidence interval \citep{geyer2005fuzzy},
\begin{equation*}
  I(\theta; \bZ) = P(\theta \in C(\bZ, w)),
\end{equation*}
where $P$ is with respect to $w \sim m(w \mid \bZ)$. For any $\theta$, the value
$I(\theta; \bz)$ is the fraction of the random confidence intervals $C(\bz, w)$,
for a fixed $\bz$, that contain $\theta$; in particular, $I(\theta; \bz) \in [0,1]$.
A value of $1$ represents values of $\theta$ that are covered by all the random intervals,
while a value of $zero$ represents values of $\theta$ covered by none. This approach has the
benefit of summarizing all the information we get into a single function.




\section{Discussion: Implications for Matching}\label{sec:disc_mat}
The as-if framework and theory we have developed is a useful tool for evaluating a given analysis method.
As an example, we consider analyzing data matched post-randomization as if it was pair randomized.
Matching is a powerful tool for extracting quasi-randomized experiments from observational studies \citep{ho2007matching, rubin2007design, stuart2010matching}.
To highlight the conceptual difficulty with post-matching analysis, we consider the idealized setting where treatment is assigned according to a known Bernoulli randomization mechanism, $\eta_0$, and matching is performed subsequently.
Specifically, units are assigned to treatment independently with probability $p_i = logit^{-1}(\alpha_0 + \alpha_1 X_i)$, where $X_i = (X_{i,1}, X_{i,2}) \in \mathbb{R}^2$ is a unit-level covariate vector. 
For simplicity, we focus on pair matching, where treated units are paired to control units with similar covariate values.
One way to analyze the pairs is as if the randomization had been a matched pairs experiment.
Although this method of analysis has already received scrutiny in the literature \citep[see, e.g., ][]{schafer2008average, stuart2010matching}, it is worth asking: is this a conditional design map with guaranteed validity?

If we can exactly match on $X_i$, then the situation is identical to that of Example~\ref{example:block}; the as-if pair randomized design map is a conditional design map, and the procedure is therefore guaranteed to be valid.
Exact matching is, however, often hard to achieve in practice.
Instead, we generally rely on approximate matching, in which the covariate distance between the units within a pair is small, but not zero.
Unfortunately, with approximate matching, the as-if pair randomized design map is not a conditional design map.
To show this formally, let $R$ be a matching algorithm which, given an assignment and fixed covariates, returns a set of $L$ pairs which we denote $M$, $M=\{(i_{1,1}, i_{1,2}), \ldots, (i_{L,1},i_{L,2})\}$.
Assuming a deterministic matching algorithm, let $M^{obs} = R(\bz)$ be the matching obtained from an observed assignment $\bz$.
Treating the matched data as a matched pairs experiment implies analyzing over all assignment vectors that permute the treatment assignment within the pairs.
A necessary condition for pairwise randomization to be a conditional procedure is that $R(\bz^*) = M^{obs}$ for all $\bz^*$ in the set of pairwise permuted assignments; that is, for any permutation of treatment within pairs, the matching algorithm $R$ must return the original matches.

This condition is not guaranteed.
To illustrate, consider the first three steps inside the light grey box in
Figure~\ref{fig:match_1}, in which we consider a greedy matching algorithm.
If we analyze the matched data as if it were a matched pairs design, then the
permutation shown is allowable by the design. However, we see in the dotted
rectangle of Figure~\ref{fig:match_1} that if we had observed that permutation
as the treatment assignment, we would have ended up matching the units differently, and therefore would have conducted a different analysis.
This is essentially the issue we encountered earlier in Section~\ref{sec:ex}; we have not created a
partition of our space.
\begin{figure}
	\centering
	\includegraphics[scale=0.58]{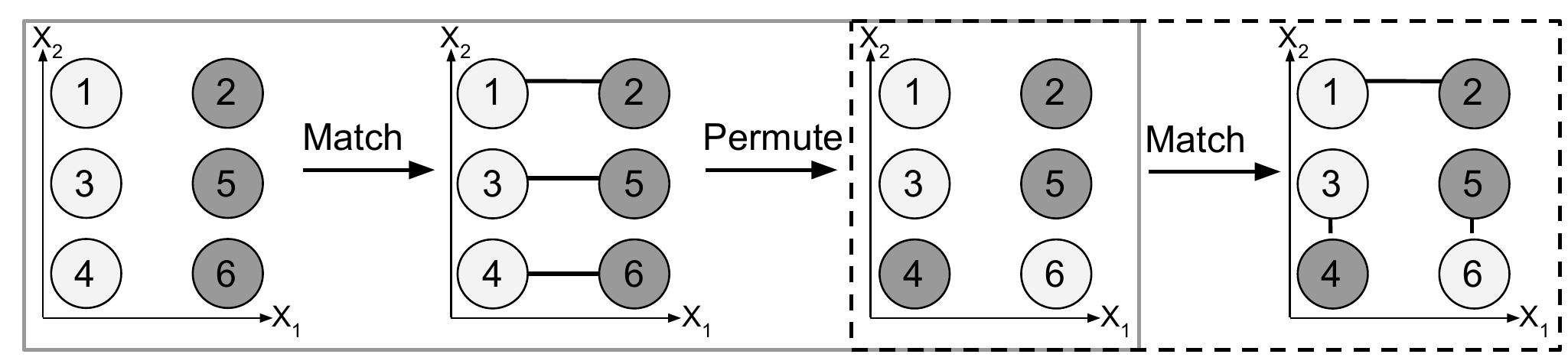}
    \caption{Units are numbered in circles. Position in the graphs corresponds to covariate values.
    Shaded circles are treated. Lines indicate matches.
    The solid edge rectangle indicates original match and one permutation.
    The  dashed edge rectangle indicates match based on that permutation.}
    \label{fig:match_1}
\end{figure}
The upshot is that when matching is not exact, analyzing the data as if it came
from a paired-randomized study cannot be justified by a conditioning argument.
A proper conditional analysis would need to take into account the matching
algorithm. Specifically, let $\mathcal{P}(M)$ be the set of all treatment assignment
vectors that are permutations of treatment within a set of matches $M$ and let
$\mathcal{V}(M)$ be the set of assignments that would lead to matches $M$ using
algorithm $R$. Then
\[\text{pr}_{\etaref}\left(\bz^* \mid R(\bz^*) = R(\bz)\right) =  \text{pr}_{\etaref}\left(\bz^* \mid \bz^* \in \mathcal{V}(R(\bz))\right).\]
Note that this is equal to $\text{Unif}\left(\mathcal{V}(R(\bz))\right)$ if assignments in $\mathcal{V}(R(\bz))$ have equal probability under the original design.
If we condition on having observed a certain matched set, we would not use all
within pair permutations of those matches, but rather would use all permutations
that would have resulted in the same matches given our matching algorithm.

Still, this distinction appears to matter more in theory than in practice.
Matching has been shown to be, in practice, a reliable tool to achieve covariate
balance such that an observational study resembles an experiment and thus gives
us hope of obtaining good inference.
Further, different inference frameworks such as those based on superpopulations with assumptions on the relationship between covariates and outcomes may provide stronger theoretical bases for matching.
See \cite{abadie2006large} for large-sample properties of matching.

\section{Discussion: Estimators, designs, and practical considerations}\label{append:est_des}
In this section we discuss some practical considerations and limitations of conditional inference and as-if analyses.
We start by discussing the impact of the choice of estimators, followed by an exploration of limitations when we do not have oracles to produce confidence intervals, and end by emphasizing difficulties regarding choices of what to condition on.

\subsection{Estimators and power}
An important consideration in practice is that the estimators we should use may change based on the design we use.
For instance, in post-stratification we would use the blocking estimator, which averages block treatment effect estimates, rather than the completely randomized simple difference estimator.
These two estimators are identical when the proportion treated within each block is the same but otherwise the completely randomized estimator will be biased for the post-stratified design.
Although this is not an issue when we have an oracle, because it will still correctly identify the distribution of $\hat{\tau}(\bZ) - \tau$ for any estimator $\hat{\tau}(\bZ)$, thus accounting for any bias, in practice we would want to use unbiased estimators.
Our validity result still holds even if we change the estimator for each different blocked design.
This validity holds because we will have conditional validity for each blocked design and therefore will have validity over the entire completely randomized design.
See proof of Theorem~\ref{corr:main_cond}, which shows that we just need conditional validity to get overall validity.
But in practice, a full comparison of a conditional and non-conditional analysis may require a comparison of relevant estimators as well.

Moreover, it is often this changing of estimators that leads to confusion about power changes for conditional analyses.
In general, a conditional analysis will not result in higher precision or smaller confidence intervals than an unconditional one, on average, if we keep our estimators the same.
This can be seen through a simple heuristic argument concerning the variance of our estimators.
Although the mode of inference in this paper does not rely strictly on the variance of our estimators, considering the variance is a proxy for considering the length of the confidence intervals.
Consider restricting randomization by conditioning on some information about the assignment given by random variable $w$.
We have the following well known variance decomposition: $\text{var}_{\etaref}(\hat{\tau}) = E_{\etaref}[\text{var}_{\etaref}(\hat{\tau}\mid w)]+\text{var}_{\etaref}(E_{\etaref}[\hat{\tau}\mid w])$.
It is easy to see that if\footnote{Note that unbiasedness of the conditional expectation is not guaranteed or required, so long as a valid analysis for the conditional design is possible. For instance, conditional unbiasedness may not hold for the conditional design map discussed in Section~\ref{sec:stoch_dm}. In these cases, we may actually see on average precision gains for the conditional designs using the same estimator as the unconditional design (if that estimator is unbiased under the unconditional design). Given the practical difficulties in conducting a valid analysis of a biased estimator, precision gains may more realistically be obtained through changes in estimator associated with the design.} $E_{\etaref}[\hat{\tau}\mid w] = \tau$ then $\text{var}_{\etaref}(\hat{\tau}) = E_{\etaref}[\text{var}_{\etaref}(\hat{\tau}\mid w)]$. 
This implies that on average, over the distribution of $w$, the conditional, or restricted, variance is equal to the unconditional, or unrestricted, variance.
This argument was made by \cite{sundberg2003conditional} under a predictive view.
\cite{sundberg2003conditional} argued for the use of conditioning for lowering the mean-squared error of the predicted squared error for variance estimation.

Of course, if we allow different estimators to be used for the conditional and unconditional designs, we do not have such guarantees and in fact may see a precision gain through changing the estimator.
For instance, in post-stratification we may see gains by using the stratified adjusted blocking estimator rather than the simple difference in means estimator.
In particular, consider two estimators $\hat{\tau}^*$ and $\hat{\tau}^\dagger$ which both have similar conditional variances such that $E_{\etaref}[\text{var}_{\etaref}(\hat{\tau}^*\mid w)] = E_{\etaref}[\text{var}_{\etaref}(\hat{\tau}^\dagger\mid w)]$.
But if $\hat{\tau}^*$ is conditionally unbiased ($E_{\etaref}[\hat{\tau}\mid w] = 0$ for all $w$) and $\hat{\tau}^\dagger$ is conditionally biased with the bias depending on $w$ ($E_{\etaref}[\hat{\tau}\mid w] = f(w)$), then we will have $E_{\etaref}[\text{var}_{\etaref}(\hat{\tau}\mid w)] = \text{var}_{\etaref}(\hat{\tau}^*) < \text{var}_{\etaref}(\hat{\tau}^\dagger)$ because $0 = \text{var}_{\etaref}(E_{\etaref}[\hat{\tau}^*\mid w])< \text{var}_{\etaref}(E_{\etaref}[\hat{\tau}^\dagger\mid w])$.

\subsection{Inference without oracles}
Inference will necessarily be more complicated without oracles.
An obvious challenge in the absence of an oracle is that we may not have good analytical control over the randomization distribution of $\hat{\tau}(\bZ) - \tau$ in the sense of being able to derive and estimate quantities such as its variance.
For instance, we may have good control over the distribution of $\hat{\tau}(\bZ) - \tau$ under $\eta_0$ but not under conditional design $\eta_{w(\bZ)}$, so that even though Theorem~\ref{corr:main_cond} guarantees the theoretical validity of this conditional as-if analysis, there is no way to implement it in practice.
Therefore, we will often want to choose an estimator that has estimable variance and is unbiasedness under the conditional distribution $\eta_{w(\bZ)}$.

A more subtle issue is that we will typically need to estimate variance in practice to construct confidence intervals and we often only have conservative estimators of the variance.
In fact, estimators of the variance may often be even more conservative for conditional analyses than unconditional analyses.
There may also be degrees of freedom losses under a conditional analysis, which can occur in post-stratification, for example.
Thus, we can find situations in which a conditional interval is smaller
than the marginal interval under the oracle, but the estimated conditional
interval is larger than the marginal conditional interval. Therefore, in
practice the theoretical benefits of conditioning must be weighed against the
drawbacks of excessively conservative variance estimation.

Three examples of conditional as-if analyses that we have identified as currently feasible without oracles are (i) analyzing a Bernoulli randomized experiment as completely randomized (Example~\ref{example:bern-crd}), (ii) analyzing a completely randomized experiment as blocked randomized (Example~\ref{example:block}), and (iii) analyzing a marginally randomized experiment with two treatments as a factorial (Example~\ref{example:fac}).
These may be simple examples but they are still pertinent; for example, online experiments (whether A/B tests run by companies or researchers on platforms such as Qualtrics) often employ Bernoulli randomization.
There are some practical limitations, even for these designs. As an example, one cannot use the standard completely randomized variance estimator for a Bernoulli assignment with 1 treated unit and 99 control units. However, even the unconditional analysis would be fraught at best in this case. Further, this would occur with low probability if assignment probability is close to 0.5. This problem becomes more salient with post-stratification, when strata can be made very fine, as discussed in the following section.

\subsection{What to condition on}

Other questions of practical importance when discussing conditioning are what and how much to condition on.
The theory indicates that more is always better; the more variables we condition on, the more relevant our analysis becomes.
But, as we condition on more information, the partition induced on the set of all assignments becomes increasingly fine.
In the extreme case, each element of the partition may contain a single assignment.
This poses a philosophical problem from the point of view of randomization-based inference because there is no randomness left!
This is not a problem under an oracle, which would give us a single point at the true $\tau$ in this case, but of course this precludes any analysis in practice.

If we cannot condition on everything, we must choose what to condition on, given
multiple options. This is analogous to the well known fact in the frequentist
literature that many ancillary statistics may exist in a given inference problem,
and that two ancillary statistics may cease to be ancillary if taken jointly
\citep{basu1964recovery,GhoshM2010ASAR}. Similarly, one must choose how fine to
make the conditioning. For instance, in our stochastic design map example of
Section~\ref{sec:ex}, one must choose the bandwidth $c$ of covariate closeness
to use. The simple but vague answer is that one should condition on the quantity,
or few quantities, that affect relevance the most, while remaining mindful of the
analytical issues described previously.
This issue of what to condition on is especially salient in the use of post-stratification, when researchers must decide what variables to use to create strata and how fine those strata should be made.
This is a fundamental statistical problem that is increasingly being highlighted in areas such as high-dimensional inference, where choices must be made in terms of which covariates to use.
Even running a simple regression requires these difficult choices to be made when the number of covariates is high compared to the number of observations.
See \cite{liu2016there}
for a more detailed discussion on the trade-offs of relevance and robustness.

\section{Conclusions}

In providing a justification for as-if analyses, our goal was not to
contradict the injunction ``analyze  the way you randomize'', but to complement
it. Our theory suggests that one can, and in fact should, analyze an experiment
in a way that is both compatible with the original randomization and also relevant to
the observed data. As we have argued, this is in essence how the conditionality
principle manifests itself in the context of randomization-based inference. 
Our line of argumentation in this paper has primarily been theoretical and
conceptual rather than the practical.

Although theoretically justified -- and, in many cases, even desirable -- valid as-if analyses may be impractical, as
discussed in Section~\ref{append:est_des}.
Practicalities currently limit the scope of conditional analyses to those we already have tools for, such as post-stratification and complete randomization, and make the benefits of conditional analyses more modest.
Despite these limitations, this paper contributes to the causal literature by
(i) aiding causal methodologists in critiquing and reflecting on complex inference strategies that exist (e.g., matching), (ii) providing a framework for developing new methods (e.g., helping to analyze data from more complicated randomization schemes, such as in \cite{basse2018exact}) and (iii) providing a framework for understanding why some styles of inference may be preferable (e.g., conditional standard errors for post stratification).


\section*{Acknowledgements}
The authors would like to thank Zach Branson, Tirthankar Dasgupta, Kosuke Imai, and participants of the 2019 Atlantic Causal Inference Conference in Montreal, Quebec as well as the Miratrix C.A.R.E.S. Lab for useful feedback.
Nicole Pashley was supported by the National Science Foundation Graduate Research Fellowship under Grant No. DGE1745303, while working on this paper.
The research reported here was also partially supported by the Institute of Education Sciences, U.S. Department of Education, through Grant R305D150040.
The opinions expressed are those of the authors and do not necessarily represent views of the Institute, the U.S. Department of Education, or the National Science Foundation.

\bibliographystyle{apalike}
\bibliography{ref}

\appendix
\begin{center}
{\bf \Large  Appendices}

\vspace{0.5cm}
\end{center}
\section{Proofs}\label{append:sec_proof}
\begin{proof}[Proof of Theorem~\ref{corr:main_cond}]
Let $\Omega$ be the set of values that $w(\bZ)$ can take for $\bZ \in \Z_{\etaref}$. Let $H(Z) = \eta_{w(\bZ)}$.
Then
	\begin{flalign*}
		P_{\etaref}\bigg\{\tau \in \Cz{H(\bZ)} \bigg\} 
        &= \sum_{\bz \in \Z_{\etaref}} P_{\etaref}(\bZ = \bz) \iv\{\tau \in \Czln{H(\bz)}\} \\
        &= \sum_{\bz \in \Z_{\etaref}} \sum_{\upsilon \in \Omega}P_{\etaref}(\bZ =\bz|w(\bZ) = \upsilon)P_{\etaref}( w(\bZ) = \upsilon) \iv\{\tau \in \Czln{H(\bz)}\} \\
        &= \sum_{\upsilon \in \Omega}P_{\etaref}(w(\bZ) = \upsilon)\sum_{\bz: w(\bz) = \upsilon} P_{\etaref}(\bZ =\bz|w(\bZ) = \upsilon)\iv\{\tau \in \Czln{\eta_\upsilon}\} \\
       &= \sum_{\upsilon \in \Omega}P_{\etaref}(w(\bZ) = \upsilon)\sum_{\bz: w(\bz) = \upsilon}P_{\eta_\upsilon}(\bZ=\bz)\iv\{\tau \in \Czln{\eta_\upsilon}\} \\
       &\geq \sum_{\upsilon \in \Omega}P_{\etaref}(w(\bZ) = \upsilon)\gamma \\
       &\geq \gamma
	\end{flalign*}
    which concludes the proof.
The second to last step comes from Proposition~\ref{prop:nulldesignmap}:

\begin{equation*}
    \sum_{\bz: w(\bz) = \upsilon}P_{\eta_\upsilon}(\bZ=\bz)\iv\{\tau \in \Czln{\eta_\upsilon}\} =P_{\eta_{\upsilon}}\bigg(\tau \in \Cz{\eta_{\upsilon}}\bigg) 
     = \gamma.
\end{equation*}

\end{proof}


\begin{proof}[Proof of Corollary~\ref{example:partition-2}]
We have $\mathcal{P} = \{\Z_{(1)}, \ldots, \Z_{(K)}\}$ a partition of the set of all possible assignments, $\Z_{\etaref}$ and
	\begin{equation*}
		w(\bZ) =\sum_k^K k \iv\{\bZ \in \Z_{(k)}\}.
	\end{equation*}
    Then
	\begin{flalign*}
	P_{\etaref}(w(\bZ) = k) &= \etaref(\Z_{(k)}).
	\end{flalign*}
	We have
	\begin{flalign*}
		P_{\etaref}(\bZ \mid w(\bZ) = k) 
		&= \begin{cases}
			\frac{\etaref(\bZ)}{\etaref(\Z_{(k)})} & \text{if $\bZ \in \Z_{(k)}$} \\
			0 &\text{ otherwise }
		\end{cases}\\
		&= \eta_k(\bZ).
	\end{flalign*}
	Hence $H: \bZ \to \eta_{w(\bZ)}$ leads to an $\etaref$-valid procedure, as a consequence of Theorem~\ref{corr:main_cond}.
\end{proof}

\begin{proof}[Proof of Unique Covariate Balance for Rerandomized As-if Design]
We have that $X_i$ is continuous for each $i$.
For example, $X_i$ could be normally distributed.
So for any random finite population and any fixed vector $a$ with $i$th entry $a_i$, such that $a_i \neq 0$ for some $i$,
\begin{align*}
P(\sum_{i=1}^n a_i X_i = 0 )= 0.
\end{align*}

The difference in covariate means (balance) is
\[\Delta_X(\bZ) = \frac{1}{N_1} \sum_{i=1}^n Z_i X_i  - \frac{1}{N-N_1} \sum_{i=1}^n (1-Z_i) X_i = \sum_{i=1}^n \left(\frac{Z_i}{N_1}-\frac{1-Z_i}{N-N_1}\right)  X_i .\]
The difference in the imbalance between assignment $\bZ$ and $\bz$ is
\[\Delta_X(\bZ) - \Delta_X(\bz) = \sum_{i=1}^n \left(\frac{Z_i-Z_i'}{N_1}-\frac{Z_i'-Z_i}{N-N_1}\right)  X_i=\left(\frac{1}{N_1}-\frac{1}{N-N_1}\right)\sum_{i=1}^n(Z_i-Z_i')   X_i. \]
Let $a_i = \left(\frac{1}{N_1}-\frac{1}{N-N_1}\right)(Z_i-Z_i')$.
Then we have 
\begin{align*}
P( \Delta_X(\bZ) - \Delta_X(\bz) = 0 )  = 0.
\end{align*}
Hence, 
\[P(\Delta_X(\bZ) - \Delta_X(\bz) = 0 \quad  \forall \bz \neq\bZ) \leq  \sum_{\bz \neq \bZ}P( \Delta_X(\bZ) - \Delta_X(\bz) = 0 ) = 0.\]
This implies that for any random finite population, the probability of any given assignment having the same covariate balance measure as another assignment is zero.
Or, in other words, with probability one each covariate balance is unique.
\end{proof}

\begin{proof}[Proof of Special Case of 0 Coverage for Rerandomized As-if Design]
We now consider the special case of Example~\ref{example:rerand}, with strict inequality for the rerandomization.
That is, we have an original design of complete randomization and design map  $H: \bZ \rightarrow P(\bz|\bz \in \mathcal{A}(\bZ))$ where $\mathcal{A}(\bZ) = \{\bz: |\Delta_X(\bz)| < | \Delta_X(\bZ)|\}$ is the set of assignments with covariate balance strictly better than the observed covariate balance. 
Consider the special case where $Y_i(0) = Y_i(1) = X_i$ ($i=1,...,N$), so $\tau = 0$.
In that case, we have $\hat{\tau}(\bZ) = \Delta_X(\bZ)$, and so for $\bZ \sim \etaref$,
\begin{equation*}
	\forall \bz \in \mathcal{A}(\bZ), \quad |\hat{\tau}(\bZ)| > |\hat{\tau}(\bz)|.
\end{equation*}
In particular, $\forall \bZ \in \Z_{\etaref}$,
$\hat{\tau}(\bZ) \not \in [\Lz{H(\bZ)}, \Uz{H(\bZ)}]$, and thus
\begin{flalign*}
	P_{\etaref}(\tau \in \Cz{H(\bZ)}) &= P_{\etaref}(0 \in \Cz{H(\bZ)})\\
    &= P_{\etaref}(\hat{\tau} - \Uz{H(\bZ)} \leq 0 
       \leq \hat{\tau} - \Lz{H(\bZ)}) \\
    &= P_{\etaref}(\Uz{H(\bZ)} \geq \hat{\tau} \geq \Lz{H(\bZ)})\\
    &= P_{\etaref}(\hat{\tau} \in [\Lz{H(\bZ)}, \Uz{H(\bZ)}]) \\
    &= 0.
    \end{flalign*}
Hence, analyzing the experiment as if it came from a stricter rerandomized design, whose non-inclusive maximal imbalance is that of the observed assignment, leads to a coverage of 0.
\end{proof}

\begin{proof}[Proof of Theorem~\ref{theorem:main_cond}]
	We have
	\begin{flalign*}
		P_{\etaref, m}\bigg\{\tau \in \Cz{H(\bZ)} \bigg\} 
		&= \int p(w) \bigg( \sum_{\bz \in \Z_{\etaref}} P_{\etaref}(\bZ=\bz \mid w) \iv\{\tau \in \Czln{H(\bz)}\}\bigg) \, dw \\
		&= \int p(w) \bigg( \sum_{\bz \in \Z_{\etaref}} \eta_w(\bz) \iv\{\tau \in \Czln{\eta_w}\}\bigg) \, dw \\
		&= \int p(w) P_{\eta_w}\bigg( \tau \in \Cz{\eta_w}\bigg) \, dw. \\
	\end{flalign*}
	Again, the key is to notice that by Proposition~\ref{prop:nulldesignmap}, $P_{\eta_w}\bigg(\tau \in \Cz{\eta_w}\bigg) \geq \gamma$.
	Hence, 
	\begin{flalign*}
		P_{\etaref, m}\bigg\{\tau \in \Cz{H(\bZ)} \bigg\} 
		&= \int p(w) P_{\eta_w}\bigg( \tau \in \Cz{\eta_w}\bigg) \, dw \\
		&\geq \int p(w) (\gamma) \, dw \\
		&= \gamma,
	\end{flalign*}
	which concludes the proof.
\end{proof}

\section{A simple example of relevance}\label{append:simp_rel}
\subsection{Relevance and betting}
First we review how the concepts of validity and relevance as formalized by \cite{buehler1959some} 
and \cite{robinson1979conditional} is made visual by the concept of betting. We follow broadly 
the setup of \cite{buehler1959some}, with some modifications to accommodate our notation.
Consider a betting game between two players:
\begin{enumerate}
	\item Player 1 chooses a distribution $\eta^\ast$ and a confidence level $\beta$ for the confidence procedure. Intuitively, this 
	corresponds to the claim $P(\tau \in \Cz{\eta^\ast}) = \beta$. 
	
	\item Player 2 selects two disjoint subsets of  $\Z$, $A^+$ and $A^-$.
	If $\bZ \in A^+$, Player 2 will bet that the confidence interval captured $\tau$.
	If $\bZ \in A^-$, Player 2 will bet that the confidence interval did not capture $\tau$.
	If $\bZ \notin A^+ \cup A^-$, Player 2 does not make a bet.
	We denote this strategy $S(A^+, A^-)$ and mathematically define it as follows:
	\begin{itemize}
		\item If $\bZ \in A^+$, bet that $\tau \in \Cz{\eta^\ast}$, betting $\beta$ to win $1$.
		\item If $\bZ \in A^-$, bet that $\tau \not \in \Cz{\eta^\ast}$, betting $1-\beta$ to win $1$.
	\end{itemize}
\end{enumerate}
The return of this game, for Player 2, is
\begin{equation*}
	R = \begin{cases}
		\iv\{\tau \in \Cz{\eta^\ast}\} (1-\beta) - \iv\{\tau \not \in \Cz{\eta^\ast}\} \beta &\text{ if } \quad Z \in A^+\\
		\iv\{\tau \not \in \Cz{\eta^\ast}\} \beta - \iv\{\tau \in \Cz{\eta^\ast}\} (1-\beta) &\text{ if } \quad Z \in A^- .\\
	\end{cases}
\end{equation*}
and the expected return is obtained by integrating over the design $\eta$. Define 
\begin{equation*}
	\beta^+ = P(\tau \in \Cz{\eta^\ast} \mid Z \in A^+)
\end{equation*}
and 
\begin{equation*}
	\beta^- = P(\tau \in \Cz{\eta^\ast} \mid Z \in A^-).
\end{equation*}
We then have
\begin{flalign*}
	E(R) &= \{\beta^+ (1-\beta) - (1-\beta^+) \beta\} P(Z \in A^+) + \{(1-\beta^-)\beta - \beta^- (1-\beta)\} P(Z \in A^-) \\
	&= \{\beta^+ - \beta\} P(Z \in A^+) + \{\beta - \beta^-\} P(Z \in A^-) .\\
\end{flalign*}
We can cast the validity criteria defined in Section~\ref{subsec:relevance} in terms of this betting (\cite{buehler1959some} uses 
the term \textit{weak exactness}). Consider the strategy $S(\Z, \emptyset)$, where $\Z$ is 
the set of all assignments. Clearly, $P(Z \in A^-) = P(Z \in \emptyset) = 0$ and $P(Z \in A^+) = P(Z \in \Z) = 1$. Moreover, $\beta^+ = P(\tau \in \Cz{\eta^\ast} \mid Z \in A^+) = P(\tau \in \Cz{\eta^\ast})$, so
\begin{equation*}
	E(R) = \beta^+ - \beta = P(\tau \in \Cz{\eta^\ast}) - \beta.
\end{equation*}
It is then easy to see that the expected return for this strategy is null if and only if the confidence interval of 
Player 1 has the advertised coverage $\beta$,
\begin{equation*}
	E(R) = 0 \quad \Leftrightarrow \quad P(\tau \in \Cz{\eta^\ast}) = \beta,
\end{equation*}
which corresponds to our frequentist strict  validity criterion. We state the following proposition:
\begin{proposition}\label{prop:validity-equiv}
	The following assertions are equivalent:
	\begin{enumerate}
		\item The procedure $\Cz{\eta^\ast}$ is strictly valid, in the frequentist sense.
		\item The expected return of strategy $S(\Z, \emptyset)$ is zero.
		\item The expected return of strategy $S(\emptyset, \Z)$ is zero.
	\end{enumerate}
\end{proposition}
In order to better understand Proposition~\ref{prop:validity-equiv}, suppose that the choice of design $\eta^\ast$ 
leads to an interval with coverage below the advertised level $\beta$. That is, such that
\begin{equation*}
	P(\tau \in \Cz{\eta^\ast}) < \beta.
\end{equation*}
Now, under the strategy $S(\emptyset, \Z)$ we have
\begin{flalign*}
	E(R) = \beta - P(\tau \in \Cz{\eta^\ast}) > 0
\end{flalign*}
and so Player 2 can make money by betting against Player 1. 
Similarly, it is easy to verify that if the interval has coverage
greater than advertised, the strategy $S(\Z, \emptyset)$ has positive return. The general idea here is 
that if Player 1 truly believes his confidence assertion, he should be willing to play against the strategies 
$S(\emptyset, \Z)$ or $S(\Z, \emptyset)$.

One insight from the betting perspective is that ensuring the strategies $S(\Z, \emptyset)$ and $S(\emptyset, \Z)$ have zero expected return may not be stringent enough a criterion for procedures.
Suppose that the design $\eta^\ast$ is such that the procedure $\Cz{\eta^\ast}$ has coverage probability $\beta$, as advertised.
Now suppose that there is a set of assignments $A$ such that
\begin{equation*}
	P(\tau \in \Cz{\eta^\ast} \mid Z \in A) = \beta^\ast > \beta.
\end{equation*}
The key question is the following: if the observed assignment $Z \in A$, should Player 1 report the confidence $\beta^\ast$ or $\beta$?
This is equivalent to asking, if you have information based on your observed assignment that you should expect better or worse coverage using the standard method, should you use that information to determine your actual confidence level?
The betting framework offers one perspective on the problem. 
Consider the strategy $S(A, \emptyset)$. The expected return is
\begin{equation*}
	E(R) = (\beta^\ast - \beta) P(\bZ \in A) > 0.
\end{equation*}
So the strategy $S(A, \emptyset)$ leads to a positive expected gain, and Player 2 can make money off of 
Player 1 by exploiting this strategy.
Following the nomenclature in \cite{buehler1959some}, we give a definition of relevance.
\begin{definition}
	A subset of assignment $A \subset \Z$ is said to be {\em relevant} for the pair 
	$\{\Cz{\eta^\ast_Z}, \beta\}$ if either $S(A, \emptyset)$ or $S(\emptyset, A)$ have non-zero expected revenue. More generally a strategy $S(A^+, A^-)$ is said to be {\em relevant} if it has non-zero expected 
	revenue.
\end{definition}

This suggests that the existence of relevant strategies against $\{\Cz{\eta^\ast_Z}, \beta\}$ is problematic 
for the procedure, even if the procedure is valid.
In fact, the notion of relevance relates to that of conditional validity.
\begin{proposition}
	The set $A$ is not relevant for $\{\Cz{\eta^\ast_Z}, \beta\}$ if and only if the procedure is valid 
	conditionally on $A$. That, is
	\begin{equation*}
	E_{S(A,\emptyset)}(R) = 0 \quad \Leftrightarrow \quad E_{S(\emptyset, A)}(R) = 0 \quad 
	\Leftrightarrow \quad P(\tau \in \Cz{\eta^\ast_Z} \mid \bZ \in A) = \beta.
	\end{equation*}
\end{proposition}

\subsection{Simple example}
Consider the usual causal inference setup with $N$ units, assuming SUTVA.
The design $\eta_0$ is Bernoulli, excluding $\bZ = (1,\dots,1)$ and $\bZ = (0,\dots,0)$.
We consider the difference in means estimator $\hat{\tau}$.
We assume a constant additive treatment effect such that for all $i$, $Y_i(1) - Y_i(0) = \tau$.
Now, conditional on $N_1 = k$, this is a completely randomized experiment, and $E(\hat{\tau} \mid N_1 = k) = \tau$.
The well known result for the variance of a completely randomized design yields
\begin{flalign*}
	\text{Var}(\hat{\tau} \mid N_1 = k) &= V^\ast \bigg( \frac{1}{k} + \frac{1}{N-k}\bigg) \\
	&\equiv v(k),
\end{flalign*}
where $V^\ast$ is the sample variance of the potential outcomes under one treatment, which is the same as under the other treatment due to the additive treatment effect.
The estimator is also unbiased unconditionally, $E(\hat{\tau}) = \tau$, but the unconditional variance becomes
\begin{flalign*}
	\text{Var}(\hat{\tau}) &= E[Var(\hat{\tau} \mid N_1)] + \text{Var}( E[\hat{\tau} \mid N_1]) \\
	&= E\bigg[V^\ast \bigg(\frac{1}{N_1} + \frac{1}{N_0}\bigg)\bigg] \\
	&= V^\ast E\bigg( \frac{1}{N_1} + \frac{1}{N-N_1}\bigg)\\
	&\equiv V,
\end{flalign*}
where the expectation is with respect to the distribution of $N_1$ induced by the design $\eta_0$.
We assume that we've reached the asymptotic regime, and so $\hat{\tau} \sim \mathcal{N}(\tau, V)$, where randomness is induced by design $\eta_0$. See Section 6 of \cite{li2017general} for an argument for why the CLT holds here.
Consider the following confidence intervals:
\begin{equation*}\label{eq:ci}
	\Cz{\eta_0} = [\hat{\tau} - 1.96 \sqrt{V}, \hat{\tau} + 1.96 \sqrt{V}].
\end{equation*}
By construction, we have $P(\tau \in  \Cz{\eta_0}) = \beta = 0.95$.
There exist winning betting strategies against this interval.
Define $\mathcal{K} = \{k: v(k) < V\}$.
Note that we have
\begin{equation*}
	v(k) < V \quad \Leftrightarrow \quad \frac{1}{k} + \frac{1}{N-k} < E\bigg( \frac{1}{N_1} + \frac{1}{N-N_1}\bigg),
\end{equation*}
which doesn't depend on $V^\ast$. So $\mathcal{K}$ doesn't depend on $V^\ast$ either.
Define $A = \{Z: N_1 \in \mathcal{K}\}$, and $A_k = \{Z: N_1 = k\}$.
Now notice that
\begin{flalign*}
	P(\tau \in \Cz{\eta_0} | N_1 = k ) &= P( -1.96 \leq \frac{\hat{\tau} - \tau}{\sqrt{v(k)}} \sqrt{\frac{v(k)}{V}} \leq 1.96 \mid N_1 = k)
\end{flalign*}
but $\frac{\hat{\tau} - \tau}{\sqrt{v(k)}} | N_1 = k \sim \mathcal{N}(0,1)$.
And so
\begin{flalign*}
	P(\tau \in \Cz{\eta_0} | N_1 = k ) 
	&=P\left( -1.96\sqrt{\frac{V}{v(k)}} \leq \frac{\hat{\tau} - \tau}{\sqrt{v(k)}}  \leq 1.96\sqrt{\frac{V}{v(k)}} \Bigg|  N_1 = k\right)\\
	&= \Phi\left(1.96\sqrt{\frac{V}{v(k)}}\right) - \Phi\left(-1.96\sqrt{\frac{V}{v(k)}}\right)\\
	&= 1 - 2 \Phi\bigg(-1.96\sqrt{\frac{V}{v(k)}} \bigg) \\
	&\equiv \beta_k
\end{flalign*}

Now the key is to notice the following:
\begin{proposition}\label{prop:over-under-coverage}
	For all $k \in \mathcal{K}$, we have
	\begin{equation*}
		\beta_k > \beta
	\end{equation*}
	and similarly, for all $k \not \in \mathcal{K}$, we have
	\begin{equation*}
		\beta_k \leq \beta
	\end{equation*}
	where the strictness in the first equation comes for the the strictness in the definition 
	of $\mathcal{K}$.
\end{proposition}

With this in place it is easy to verify that $S(A, A^c)$ is a relevant strategy. Indeed, if
$N_1 = k \in \mathcal{K}$, we have
\begin{flalign*}
	E(R \mid \bZ \in A_k) &= P(\tau \in \Cz{\eta_0} \mid \bZ \in A_k) (1-\beta) - P(\tau \not \in \Cz{\eta_0} \mid \bZ \in A_k) \beta \\
	&= \beta_k - \beta \\
	&> 0
\end{flalign*}
and so
\begin{flalign*}
	E(R \mid \bZ \in A) &= \sum_{k \in \mathcal{K}} E(R \mid \bZ \in A_k) P(Z \in A_k | \bZ \in A) \\
	&= \sum_{k \in \mathcal{K}} (\beta_k - \beta) P(\bZ \in A_k \mid \bZ \in A) \\
	&> 0.
\end{flalign*}
Similar derivations show that $E(R \mid \bZ \in A^c) > 0$. And so
\begin{flalign*}
	E(R) &= E(R \mid \bZ \in A) P(\bZ \in A) + E(R \mid \bZ \in A^c) P(\bZ \in A^c) \\
	&> 0
\end{flalign*}
which means that the strategy $S(A, A^c)$ is relevant for the procedure $\{\Cz{\eta_0}, \beta\}$.


\end{document}